\documentclass[a4paper,12pt]{article}
\usepackage{authblk}
\usepackage{a4wide}
\usepackage{amsfonts}
\usepackage{epsfig}
\usepackage{amsmath,amssymb,latexsym}
\usepackage{enumerate}
\usepackage{graphicx}
\usepackage{amsthm}
\usepackage{xcolor}

\def\P{{\mathbf P}}
\def\N{{\mathbf N}}
\def\Q{{\mathbf Q}}

\def\H{{\mathbf H}}
\def\K{{\mathbf K}}

\def\emptyset{{\varnothing}}

\newcommand\Model[2][]{\ensuremath{R^{#1}_{#2}}}
\newcommand\Path[1]{\ensuremath{W_{#1}}}

\newtheorem{thm}{Theorem}

\newtheorem{prop}[thm]{Proposition}
\newtheorem{lem}[thm]{Lemma}

\newtheorem{property}[thm]{Property}

\newtheorem{rem}[thm]{Remark}

\newtheorem*{process}{Replacement process}

\title{On dually-CPT and strong-CPT posets}

\author[a,b]{Liliana Alc\'on}

\author[c]{Martin Charles Golumbic}

\author[a,b]{Noem\'\i \, Gudi\~{n}o}

\author[a,b]{Marisa Gutierrez}

\author[d,e]{Vincent Limouzy}

\affil[a]{Centro de Matem\'{a}tica CMaLP, FCE, Universidad Nacional de La Plata. Argentina}
\affil[b]{CONICET, Buenos Aires, Argentina}

\affil[c]{Department of Computer Science, University of Haifa,
              Mt. Carmel, Haifa, Israel}

\affil[d]{Universit\'e Clermont Auvergne, Clermont
              Auvergne INP, CNRS, Mines Saint-Etienne, \textsc{Limos}, F-63000
              Clermont-Ferrand, France}
\affil[e]{This research was supported by the ANR project GRALMECO (ANR-21-CE48-0004-01)}

\begin{document}

\maketitle

\begin{abstract}
A poset is a containment of paths in a tree (CPT)
if it admits a representation by containment where each element of
the poset is represented by a path in a tree and two elements
are comparable in the poset if and only if the corresponding paths
are related by the inclusion relation. Recently Alc\'on, Gudi\~{n}o and Gutierrez~\cite{ALC-GUD-GUT-2}
introduced proper subclasses of CPT posets, namely dually-CPT, and strongly-CPT.
A poset $\P$ is dually-CPT, if and only if $\P$ and its dual $\P^{d}$ both admit a CPT representation.
A poset $\P$ is strongly-CPT, if and only if $\P$ and all the posets that share the same
underlying comparability graph admit a CPT representation. Where as the inclusion between
Dually-CPT and CPT was known to be strict. It was raised as an open question by  Alc\'on, Gudi\~{n}o and Gutierrez~\cite{ALC-GUD-GUT-2}
whether strongly-CPT was a strict subclass of dually-CPT. We provide a proof that
both classes actually coincide.
\end{abstract}

\section{Introduction}

A poset is called a containment order of paths in a tree (CPT for short)
if it admits a representation by containment where each element of the poset
corresponds to a path in a tree and for two elements $x$ and $y$, we have
$x < y$ in the poset if and only if the path corresponding to $x$ is properly
contained in the path corresponding to $y$.

Several classes of posets are known to admit specific containment models,
for example, containment orders of circular arcs on a circle~\cite{NiMaNa88, RoUr82},
containment orders of axis-parallel boxes in $\mathbb{R}^d$ \cite{GoSc89}, or containment
orders of disks in the plane \cite{BrWi89, Fi88, Fi89} to cite just a few.
All the aforementioned classes, as well as CPT,
generalize the class {CI} of containment orders of intervals on a line~\cite{DU-MI-41}.
It is well known that this class coincides with the class
of $2$-dimensional posets and are also equivalent to the
transitive orientations of permutation graphs~\cite{Go2004}.

In 1984, Corneil and Golumbic observed that a graph $G$
may be the comparability graph of a CPT poset,
yet a different transitive orientation of $G$ may not necessarily
have a CPT representation, (see Golumbic~\cite{GO-84}).
This stands in contrast to poset dimension, interval orders,
unit interval orders, box containment orders, tolerance orders
and others which are comparability invariant.
Golumbic and Scheinerman~\cite{GoSc89}
called such classes \emph{strong containment poset classes}.

Recently, interest in CPT posets has been revived and several
groups of researchers have considered various aspects of this class
\cite{ALC-GUD-GUT-2, AGG-k-tree, GolumbicL21, MajumderMR21}.
Since the CPT posets are not a strong containment class,
Alc\'on, Gudi\~{n}o and Gutierrez~\cite{ALC-GUD-GUT-2}
introduced the study of the subclasses
dually-CPT and strongly-CPT posets.
A poset $\P$ is called \emph{dually-CPT} if $\P$ and its dual $\P^{d}$ admit
a CPT representation. A poset $\P$ is called \emph{strongly-CPT} if $\P$
and all the posets that share the same underlying comparability
graph admit CPT representations. From the definition
it is clear that the class of strongly-CPT posets is included
in the class of dually-CPT posets. Many families of separating examples
are now known between the class of dually-CPT and general CPT posets,
however, concerning the strongly and dually-CPT, it
was left as an open problem for many years to determine whether the
inclusion is strict or if the two classes coincide.

We present in this paper a solution to this question with
the following main theorem.

\begin{thm}
\label{thm:Main}
A poset $\P$ is strongly-CPT if and only if it is dually-CPT.
\end{thm}

To prove our main result we rely on the link between modular decomposition
of the underlying comparability graph and its transitive orientations.
Our strategy consists of considering a dually-CPT poset and proving
that any poset with the same comparability graph also admits a CPT representation.
At first we consider the representation and perform some
modifications to obtain a representation with particular properties. Once
this is done, we rely on the specific structure of modules in dually-CPT
posets, and we provide a method to obtain the representation of
any poset with the same comparability graph.

~\\

The paper is organized as follows: In Section \ref{sec:Def}, we present
the definitions related to posets, CPT and modular decomposition
and recall some fundamental results that we will use throughout
the paper. In Section \ref{sec:TrivPathsModules}, we prove that
for dually-CPT posets it is possible to obtain a representation
where no element of a strong module is represented by a trivial path. Then,
in Section \ref{sec:ModuleVsTrivPath}, we show how to modify a CPT representation
of a dually-CPT poset so that either the paths of a strong module
do not end on a trivial path or the considered module admits
very specific properties.
Finally, in Section~\ref{sec:Substition}, we show how to use an operation called
substitution to prove our main result.

\section{Definitions and notations}
\label{sec:Def}

A \emph{partially ordered set} or \emph{poset}  is a pair
$\P=(X,P)$ where $X$ is a  finite non-empty set
and $P$ is a reflexive, antisymmetric
and transitive binary relation on $X$. The elements of $X$ are
also called \emph{vertices} of the poset. As usual, we write $x \leq
y$ in $\P$ for $(x,y)\in P$; and $x<y$ in $\P$ when $(x,y)\in P$
and $x\neq y$. If $x<y$ or $y<x$, we say that $x$ and $y$ are
\emph{comparable} in $\P$ and write $x\perp y$. When there is no relationship between $x$ and
$y$ we say that they are
\emph{incomparable} and write $x\parallel y$.
An element $x$ is \emph{covered} by $y$ in \textbf{P},
denoted by $x<:y$ in \textbf{P}, when $x<y$ and there is no
element $z\in X$ for which $x<z$ and $z<y$.
The \emph{down-set} $\{x\in X: x< z\}$  and the \emph{up-set} $\{x\in
X : z< x\}$  of an element $z$ are denoted by $D(z)$ and $U(z)$, respectively. We let
$D[z]=D(z) \cup \{z\}$ and $U[z]=U(z) \cup \{z\}$.
The \emph{dual} of  $\P=(X,P)$ is the poset  $\P^d=(X,P^d)$
where $x\leq y$ in $\P^d$  if and only if  $y \leq x$ in $\P$.

A \emph{containment representation} $\Model{\P}$  or \emph{model} 
of a poset $\P=(X,P)$ maps
each element $x$ of $X$ into a set $\Path{x}$ in such a way that $x <
y$ in $\P$  if and only if  $W_x$ is a proper subset of $W_y$.
We identify the containment representation $\Model{\P}$ with the set family
$\{\Path{x}\}_{x\in X}$.

A poset $\P=(X,P)$ is a  \emph{containment order of paths in
a tree}, or $CPT$ poset for brevity, 
if it admits a containment representation $\Model{\P} = \{W_x\}$ where
every $W_x$ is a path of a tree $T$, which is called the
\emph{host tree} of the model. 
When $T$ is a path, $\P$ is said to be a \emph{containment order
of intervals} or $CI$ poset for short.
(We generally consider a path as the set of vertices that induces it.)

The comparability graph $G_{\P}$ of a poset $\P=(X,P)$ is the
simple graph with vertex set $V(G_{\P})=X$ and edge set
$E(G_{\P})=\{xy: x\perp y\}$. In what follows, a poset $\P$, such
that $G_{\P}$ is complete (resp. without edges), is called a \emph{total order} (resp. an
\emph{empty order}). We say that two posets are
\emph{associated} if their comparability graphs  are isomorphic. A
graph $G$ is a \emph{comparability graph} if there exists some
poset $\P$ such that $G=G_{\P}$.

A \emph{transitive orientation}
$\overrightarrow{E}$ of a graph $G$ is an assignment of one
of the two possible directions, $\overrightarrow{xy}$ or
$\overrightarrow{yx}$,  to each edge $xy\in E(G)$ in such a way that
if $\overrightarrow{xy}\in \overrightarrow{E}$ and
$\overrightarrow{yz}\in \overrightarrow{E}$ then
$\overrightarrow{xz}\in \overrightarrow{E}$. The graphs whose
edges can be transitively oriented are exactly the comparability
graphs \cite{GA-67, GH-HO-62, Go2004}. Furthermore, given a transitive
orientation $\overrightarrow{E}$ of a graph $G$, we let
$\P_{\overrightarrow{E}}$  denote the poset
$(V(G),P_{\overrightarrow{E}})$ where $u<v$ in
$\P_{\overrightarrow{E}}$ if and only if $\overrightarrow{uv}\in
\overrightarrow{E}$. The comparability graph of
$\P_{\overrightarrow{E}}$ is $G$. Thereby, the transitive
orientations of $G$ are put in one-to-one correspondence with the posets whose comparability
graphs are $G$.

Let $\P=(X, P)$ be a poset. A set $M\subseteq X$ is a \emph{module}
(\emph{homogeneous set} \cite{GA-67})
if for every $y\in X-M$, either  $y\perp x$ for all
$x\in M$, or $y\parallel x$ for all $x \in M$.
The whole set $X $ and
the singleton sets $\left\{x\right\}$, for any $x\in X$,
are modules of $\P$. These modules are called \emph{trivial modules}.
A poset $\P$ is \emph{prime} or \emph{indecomposable} if all its modules are trivial. Otherwise
$\P$ is \emph{decomposable} or \emph{degenerate}.
A module $M$ is \emph{strong} if for all modules $M'$
either $M\cap M'=\emptyset$ or
$M\subseteq M'$ or $M'\subseteq M$.

A module (respectively, strong module) $M\neq X$ is called \emph{maximal}
if there exists no module (respectively, strong module)
$Y$ such that $M\subset Y \subset X$.

\begin{thm}(Modular decomposition theorem) \cite{GA-67}
Let $\P=(X, P)$ be a poset with at least two vertices. Then exactly
one of the following three conditions is satisfied:
\begin{enumerate}
	\item [(i)] $G_{\P}$ is not connected and the maximal
	strong modules of $\P$ are the connected components of $G_{\P}$.
	
	\item [(ii)] $\overline{G_{\P}}$ is not connected and the maximal
	strong modules of $\P$ are the connected components of $\overline{G_{\P}}$.
	
		\item [(iii)] $G_{\P}$ and $\overline{G_{\P}}$ are connected. There is some
		$Y\subseteq X$ and a unique partition $\mathcal{S}$ of $X$ such
		that
		\begin{enumerate}
	\item [(a)] $|Y|\geq 4$,
	\item [(b)] $\P\left[Y\right]$ is the biggest prime
	subposet of $\P$ (in the sense that it is not included in any
	other prime subposet),
	\item [(c)]	for every part $S$ of the partition $\mathcal{S}$,
	$S$ is a module of $\P$ and $|S\cap Y|=1$.
\end{enumerate}
 \end{enumerate}
\end{thm}

The previous theorem defines a partition $\mathcal{M}(\P) = \{M_1,...,M_k\}$  of $X$,
which is  called the \emph{canonical partition} or
\emph{maximal modular partition} of $\P$.
In the first case,  $G_{\P}$ is said to be \textit{parallel} or \emph{stable} and  the partition is formed by
the vertices
of the connected components of $G_{\P}$. In the second case,  $G_{\P}$ is \emph{series} or \emph{clique} and
the partition is formed by the vertices of each connected component of $\overline{G_{\P}}$.
And, in the last case,
$G_{\P}$ is \emph{neighborhood} or \emph{prime}, and the partition is $\mathcal{S}$.

The \emph{quotient poset} of $\P$, denoted by $\P/\mathcal{M}(\P)$,  has a vertex $v_i$ for
each  part $M_i$ of $\mathcal{M}(\P)$; and two
vertices $v_i$ and $v_j$ of $\P/\mathcal{M}(\P)$ are comparable 
if and only if for all  $x\in M_i$
and for all $y\in M_j$, $x \perp y$ in $\P$.

The quotient poset  is \emph{empty} (iff $G_{\P}$ is parallel), a \emph{total order} (iff
$G_{\P}$ is series) or \emph{indecomposable} (iff $G_{\P}$ is neighborhood).

On some occasions, when referring to a module, we will mean the subposet induced by it. For
instance,
we will say that a module $M$ of $\P$ is $CI$ or that it is prime, meaning that $\P(M)$ is.
This will be clear from the context and will  cause no confusion.

\begin{thm}\cite{GA-67} \label{t:indecommposable-poset}
Given posets $\P$ and $\P'$, if $G_{\P}=G_{\P'}$ and
$\P$ is indecomposable, then $\P'=\P$ or $\P'=\P^d$.
\end{thm}

\begin{prop}\cite{GA-67} \label{p:associated-posets}
Given posets $\P$ and $\P'$, if $G_{\P}=G_{\P'}$, then $\P$ and
$\P'$ have the same strong modules and, consequently,
$\mathcal{M}(\P)=\mathcal{M}(\P')$.
\end{prop}

Given a vertex $v$ of a poset $\P=(X, P)$ and a poset $\H=(X_1, H)$, \emph{substituting or
replacing} $v$ by $\H$ in $\P$ results in the poset
$\P_{\H\rightarrow v}=\left(X-\{v\}\cup X_1, P_{\H\rightarrow v}\right)$ such that
$P_{\H\rightarrow v}=P-\{(x,y) : x=v \vee y=v\} \cup H \cup \{(x,y):x\in X_1 \wedge  y\in
U(v)\}$ $\cup \{(x,y):y\in X_1 \wedge x\in D(v)\}$.

\begin{thm} \label{t:orientations}
Let $\mathcal{M}(\P) = \{M_1, . . . , M_k\}$ be the maximal modular partition of a connected poset
$\P=(X, P)$ whose quotient is prime, and call $\H$ the
quotient poset $\P/\mathcal{M}(\P)$.  A poset $\Q$ is associated to $\P$ if and only if there
exist posets  $\Q_i$ for $1\leq i \leq k$ such that
$\Q_i$ is associated to $\P_i=\P(M_i)$ for each $i$, and $\Q$ is obtained by replacing each
vertex $v_i$ of $\H$ by the poset $\Q_i$ or replacing each vertex $v_i$ of $\H^d$ by the poset
$\Q_i$.
\end{thm}

\begin{thm}\cite{GA-67} \label{t:strong-property}
A poset $\P$ is $CI$ if and only if the quotient poset and all the maximal strong modules of
$\P$ are $CI$.
\end{thm}

\begin{lem} \cite{ALC-GUD-GUT-2}
\label{l:neces} If $z$ is a vertex of a $CPT$ poset $\P$ then the subposet
 induced by the  closed  down-set  of $z$ is  $CI$. In particular, if $\P$ is dually-$CPT$,
then also the subposet induced by the closed up-set of $z$ is $CI$.\end{lem}

\begin{rem}\cite{GA-67}
\label{r:strong}  Let  $\P$ and $\P'$ be associated posets. Then,
$\P$ is a $CI$ poset if and only if  $\P'$ is a $CI$ poset. In
particular, $\P$ is a $CI$ poset  if and only if  $\P^d$ is a $CI$
poset.
\end{rem}

\begin{thm} \label{t:dually-nprime}
Let $\P=(X, P)$ be a connected dually-$CPT$ poset. Then the quotient poset of $\P$ is
dually-$CPT$ and every maximal strong module of $\P$ is $CI$. In particular, if the quotient
poset  is $CI$, then $\P$ is $CI$.
\end{thm}
\begin{proof} Let $\mathcal{M}(\P) = \{M_1, . . . , M_k\}$ be the maximal modular partition of
$\P$. The quotient poset $\H=\P/\mathcal{M}(\P)$ is a subposet of $\P$, so  $\H$ is
dually-$CPT$. We can assume that $\P$ is not empty, and since $\P$ is connected we have that
$\H$ is connected, and so  every
vertex $v_i$ of $\H$ is in the down-set or in the up-set of some other vertex. Which implies
that in $\P$ the whole module $M_i$ is in the up-set or in the down-set of some other vertex.
It follows from  Lemma \ref{l:neces} that each $\P_i=\P(M_i)$ is $CI$. Therefore, by Theorem
\ref{t:strong-property}, if $\H$ is $CI$, then $\P$ is $CI$.\end{proof}

The converse  of Theorem \ref{t:dually-nprime} is not true in general. 
For instance, if in the quotient poset $\H$ there exists a vertex $v_i$ 
such that in any $CPT$ representation of $\H$ the corresponding
path $W_{v_i}$ is reduced
to a vertex, then for $\P$ to be $CPT$ the module $M_i$ has to be a singleton.

%
%
%
%
%
%
%
%

In a representation  $\Model{\P}$ of a CPT poset $\P$, 
a subset $X$ of paths of \Model{\P}
is called \emph{one-sided} if all the paths that represent $X$
arrive at a vertex $a$ of the host tree
and all paths of $X$, except possibly one trivial path,
pass through a vertex $b$ of $T$
neighbor of $a$. If all the paths of $X$ arrive at a vertex $a$ and $X$ is not one-sided, then it is called \emph{two-sided}.

Addressing that issue in the proof of the main theorem will requires the following lemmas and
properties.

\begin{property}\cite{DU-MI-41, GO-84}
\label{p:compresing-CI-model}
Every $CI$ poset admits a $CI$ representation
 where the intersection of all the intervals used in
the representation is a non-trivial  interval.

\end{property}

\section{Trivial paths into modules}
\label{sec:TrivPathsModules}

The goal of this section is to prove that for any dually-CPT
poset $\P$, there exists a representation $\Model{\P}$ where all
the elements contained in strong modules are represented
by non-trivial paths.

At first we prove that if an element of module
is represented by a trivial path, it does mean that the module
(all its elements) are not greater than any other element not in the module.

\begin{lem}
 Let $\P$ be a poset and let $M$ be a strong module of $\P$. If there
 exists a representation $\Model{\P}$ where an element $x$ of $M$ is represented
 by a trivial path, then all the elements of $M$ are not greater than any element
 of $\P$ not in $M$.
 \label{lem:TrivialPathModule}
\end{lem}

\begin{proof}
 Let us proceed by contradiction and let us assume that there exists an element $z \notin M$
 such that $z < x$. Then in any representation \Model{\P} we have $\Path{z}  \subset \Path{x}$
 but
 since $\Path{x}$ is already a trivial path it cannot properly contain some other object.
\end{proof}

Hence from the previous lemma, if in a representation \Model{\P} an element
of a module is represented by a trivial path, the module is a minimal subset of $\P$.

\begin{lem}
\label{lem:ModuleContained}
 Let $M$ be a strong module of a CPT poset $\P$,
 if in a representation $\Model{\P}$ one of its elements
is represented as a trivial path, then there exists an element $x$ not in $M$
such that the path $\Path{x}$ contains all the paths representing the elements of $M$.
\end{lem}

\begin{proof}
 Since the poset is connected, and by the previous lemma, we know that the
 module cannot contain any other element, to ensure the connection outside
 the module, there might be at least one element $x$ that is greater than  every element
 of $M$.
\end{proof}

\begin{lem}
\label{lem:PassingThroughA}
  Let $M$ be a strong module of a CPT poset $\P$.
  If in a representation $\Model{\P}$ one of its elements $z$
is represented as a trivial path, then this path is hosted
on some vertex $a$ of $T$. If for an element $x$ not in $M$  its path $\Path{x}$
passes through $a$, then $\Path{x}$ has to contain all the paths corresponding
to the elements of $M$.
\end{lem}

\begin{proof}
 From the definition of a module, every element not in the module is either
 completely disconnected from $M$ or completely connected to $M$. In that case,
 if for an element $x$, in a representation $\Model{\P}$ its path $\Path{x}$ passes through
 $a$,
 then it is connected to the element $z$. Hence it has to be connected to every
 element of $M$. In addition, in a transitive orientation of a graph, the containment
 relation between $x$ and the elements of $M$ is the same for every element of $M$.
 Thus if $\Path{x}$ contains $\Path{z}$ it contains all the paths of the elements of $M$.
\end{proof}

\begin{lem}
\label{lem:non-triv}
 Let $M$ be a strong module of a CPT poset $\P$. If in a representation
 $\Model{\P}$ one of its elements
is represented as a trivial path and $M$ is a clique or prime module,
then there exists at least one element of $M$ represented as a non-trivial path.
\end{lem}

In the case of dually-CPT posets, the next three lemmas consider the presence
of trivial paths in a representation of strong modules and  show how
to obtain an equivalent representation where all the elements of the module
are non-trivial paths. For these lemmas, we consider each strong module to
be a CI poset and the element of the module represented by a trivial
path is denoted by $z$.

\begin{lem}
\label{lem:CliqueModuleTrivialPath}
 Let $M$ be a strong CI clique module of a dually-CPT poset $\P$.
 If an element $z$ of $M$ is represented
 by a trivial path in a representation $\Model{\P}$, then there exists a representation
 $\Model[']{\P}$ where $z$ is represented as non-trivial path.
 \end{lem}

 \begin{proof}
    By Lemma \ref{lem:ModuleContained} we know that there exists an element $x$
    such that all the paths of $M$ are contained in $\Path{x}$ in all CPT representations.
  Let us consider three cases.

  ~
  \\(1) Suppose the trivial path of $z$ is not an extremity
  of any path that represents the elements of $M$. Let $a$ be the vertex of $T$
  that hosts the trivial path of $z$. Since $\Path{z}$ is not an extremity of any path
  of $M$, $a$ admits at least one neighbor $b$ in $T$ such that all the paths
  of $M$ (except for $z$) pass  through $b$ (see Figure \ref{fig:CliqueModule}$(i)$).
  Let us subdivide the edge $a,b$ by adding a vertex $c$.
  Then it suffices to replace the trivial path of $z$ by a non-trivial path that goes
  from $c$ to $a$ in $T$.  The containment relations among $M$ are preserved
  and no new containment relation is added nor deleted with respect to the elements
  not in $M$.

  ~
  \\(2) Suppose now, the trivial path of $z$ is a common extremity for all the
  elements of $M$ and $M$ is one-sided (see Figure \ref{fig:CliqueModule}$(ii)$).
  We proceed as in the previous case; we consider a vertex $b$ of $T$ that is a neighbor
  of $a$ and such that all the paths of $M$ except for $z$ pass through $b$.
  Since $M$ is a clique, it only admits at most one element represented by
  a trivial path, such a vertex $b$ exists, then we subdivide the edge by adding
  a vertex $c$ and the path of $z$ goes from $a$ to $b$.
  Note that the technique still works if some paths of $M$ continue after $a$.

  ~
  \\(3) Suppose now, the trivial path of $z$ is the common extremity for
  some paths of the module in a 2-sided manner (see Figure \ref{fig:CliqueModule}$(iii)$).
  Let $b$ and $c$ be two vertices of $T$ that are neighbors of $a$, such that  $b$ and $c$ lie
  on the path of $x$, $x$ being an element not in $M$ that contains all elements of $M$.
  We can partition the elements of $M$  into three sets:
   $B$ the elements for which paths arrive at $a$ and pass through $b$,
    $C$  defined in a similar way but \emph{w.r.t.} $c$ instead of  $b$,
  and $A$, the paths of $M$ that go through $b$ and $c$. This time we need
  to subdivide the edges $a,b$ and $a,c$ of $T$. We add a vertex $i$ between $a$ and
  $b$ and a vertex $j$ between $a$ and $c$. Then it suffices to extend the paths of $B$ until
  $j$ and the paths of $C$ until $i$. The path of $z$ now goes from $i$ to $j$. By subdividing
  several times the edges $a,b$ and $a,c$, we can make sure that all
  the extremities are distinct.
 \end{proof}

\begin{figure}
 \begin{center}
  \includegraphics[width=\textwidth]{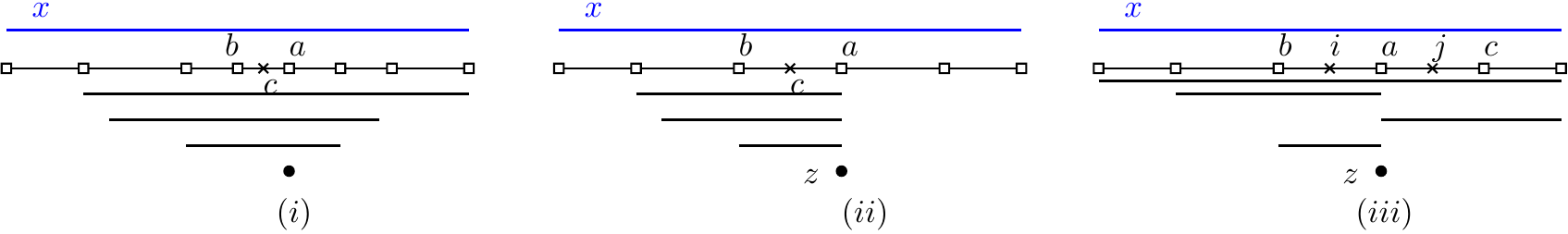}
 \end{center}
\caption{Representation of cliques modules with trivial paths.}
\label{fig:CliqueModule}
\end{figure}

\begin{lem}

\label{lem:StableModuleTrivialPath}
 Let $M$ be a strong CI stable module of a dually-CPT poset $\P$.
 If an element $z$ of $M$ is represented
 by a trivial path in a representation $\Model{\P}$, then there exists a representation
 $\Model[']{\P}$ where $z$ is represented as non-trivial path.
 \end{lem}

 \begin{proof}
 Let us first remark that in a strong stable module, several elements can be
 represented as trivial paths in a representation $\Model[']{\P}$. In addition,
 if an element of $M$  is represented by a trivial path, the trivial
 path is disjoint from all the other paths representing the elements of $M$.
 Let $z$ be such an element.
 We will transform $\Model{\P}$ such that all the elements of $M$ represented by trivial paths
 in $\Model{\P}$ will be represented by non-trivial paths. Let $a$ be the vertex of $T$
 that hosts the path of $z$. Thanks to Lemma \ref{lem:ModuleContained}, we know
 that there exists an element $x$ of $\P$
 such that in $\Model{\P}$ the paths of the elements of $M$
 are contained in the path of $x$. Since $M$ is a non-trivial module it contains
 at least two elements, hence in $\Model{\P}$ there exists a vertex $b$ of $T$ that is adjacent
 to $a$, and $b$ is contained in all the paths of the elements not in $M$ that
 contain $M$, since such a path has to contain $\Path{z}$ and all the other elements of $M$.

 Let us denote by $U=\{u_1,u_2,\ldots,u_k\}$ the elements of $M$ that are represented
 by trivial paths in $\Model{\P}$.
 To obtain an equivalent representation $\Model[']{\P}$, we  subdivide $2k-1$ times the edge
 $a,b$.
 We then rename $a$ as $a_1$, and we
 number the newly created vertices $a_2,a_3,\ldots,a_{2k}$ (the transformation
 is presented in Figure \ref{fig:StableModule}). In this new
 representation each element $u_i$ of $U$ is replaced by a path
 that goes  from $a_{i}$ to $a_{k+i}$ in $T$.

 It remains to prove that this representation is equivalent. First observe that
 for any element $x$ connected to $M$, its path in $\Model{\P}$ contains all the
 elements of $M$. By the choice of vertex $b$ to perform the transformation, we
 can guarantee that any path of such an element $x$ will pass through $a,b$ in $\Model{\P}$.
 Since we subdivided this edge to obtain $\Model[']{\P}$ , this path will still pass through
 $a$ and $b$
and all the vertices introduced by the transformation.

Now for any element $y$ not connected to $M$, we know by
Lemma \ref{lem:PassingThroughA} that no path of such an element will pass through $a$.
\end{proof}

\begin{figure}
 \begin{center}
  \includegraphics{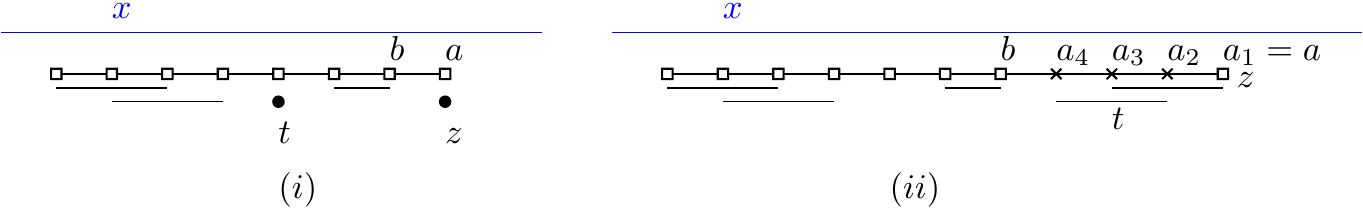}
 \end{center}
 \caption{$(i)$ Representation of a stable module with elements represented
 by trivial paths; $(ii)$ transformation to eliminate trivial paths
 from the representation.}
 \label{fig:StableModule}
\end{figure}

\begin{lem}
\label{lem:PrimeModuleTrivialPath}
Let $M$ be a strong CI prime module of a dually-CPT poset $\P$.
If an element $z$ of $M$ is represented
by a trivial path in a representation $\Model{\P}$, then there exists a representation
$\Model[']{\P}$ where $z$ is represented as non-trivial path.
\end{lem}

\begin{proof}
 For this proof, we consider three cases: (1) either $\Path{z}$ the trivial path of $z$ is
 properly contained (\emph{i.e.} $\Path{z}$ is not an extremity of
 any path of the element of $M$) in all the paths of the elements of the module $M$, or (2)
 there exists at least two elements $q$ and $r$ of $M$ such that $\Path{z}$ is the right bound
 of $\Path{q}$ and the left bound of $\Path{r}$, or (3) the path $\Path{z}$ is the right
 (respectively left) bound for some paths  representing elements of $M$,
 and is not  the left (respectively right) bound of any elements of $M$.
 These three cases are illustrated in Figure
 \ref{fig:PrimeModuleTrivialPath}$(i)-(iii)$.

~
\\ \noindent
 { (1)}
 Let $a$ be the vertex of $T$ that hosts $\Path{z}$, the trivial path representing $z$.
 By hypothesis, all the paths that represent the elements of $M$ properly contain
$\Path{z}$ and thus pass through vertex $a$. Since it is a proper containment, no
path of elements of $M$ (other than $z$) starts or finishes at $a$. Thus
$a$ admits at least one neighbor $b$ in $T$ such that all the paths
that represent elements of $M$, except for $z$, pass through $b$. To obtain
a  new representation $\Model[']{\P}$ we subdivide the edge $a,b$ by a adding
a vertex $d$. Then $\Path{z}$ in $\Model[']{\P}$ is replaced by the path $a,d$.
(See Figure \ref{fig:PrimeModuleTrivialPath}$(iv)$).
Since the representation of $\Path{z}$ is the only modification of the
representation, by the previous discussion all the paths that represent the
elements of $M$ pass through $a$ and $b$ and as a consequence pass
through $a$ and $d$ since $d$ is in between $a$ and $b$.
By Lemma \ref{lem:PassingThroughA} we know that all the paths of the elements not in $M$ that
pass through  $a$ will also contain all the paths of $M$. Hence the modification of $\Path{z}$
preserves the containment relation of $\Model{\P}$.

~
\\ \noindent
{(2)}
Let us now consider that there exist at least two elements $q$ and $r$ of $M$
such that in $\Model{\P}$, the vertex $a$ is the right bound of
the path $\Path{q}$ and the left bound of the path $\Path{r}$ (see Figure
\ref{fig:PrimeModuleTrivialPath}$(ii)$).
Let us denote by $L$ the set of elements of $M$ for which $a$ is the right bound
in the representation $\Model{\P}$ and similarly let us denote by $R$ the set of
elements of $M$ for which $a$ is the left bound in $\Model{\P}$. Let us remark
that $L\cap R = \emptyset$ and some elements of $M \setminus (L\cup R)$ might
not be empty.  Let $b$ be the neighbor of $a$ in $T$ such that the paths of
the elements of $L$ pass through $b$. And similarly let $c$ be the neighbor
of $a$ in $T$ such that the paths of the elements of $R$ pass through $c$.
To obtain a new representation $R'(P)$ we subdivide  the edge $a,b$ $|R|+1$ times,
and the edge $a,c$ $|L|+1$ times. The added vertices are called $ab_{i}$ for the vertices
between $a$ and $b$ and $ac_{j}$ for the vertices between $a$ and $c$.
Let $ab_{1}$ and $ac_{1}$ be the neighbors of $a$
in $\Model[']{\P}$. The path  $\Path{z}$ now goes from $ab_{1}$ to $ac_{1}$. 
The left bound
of the paths of the elements of $R$ are moved on the $ab_{i}$ vertices. The coordinates
are chosen to preserve the containment relation. We proceed symmetrically
for the paths of the elements in $L$. It remains to prove that the obtained representation
still corresponds to $\P$. Again we know by Lemma \ref{lem:PassingThroughA} that no
path of an element not connected to $M$ passes through $a$, by construction
it remains valid for $a$ and for all the newly introduced vertices. For any other path
their relation to $\Path{z}$ and the paths of the elements of $L$ and $R$ are unchanged.
If the path $\Path{s}$ of an element $s$ was containing the path $\Path{l}$ of an element
$l$ of $L$ in $\Model{\P}$, it is still the case in $\Model[']{\P}$. In that case the left
bound
of $\Path{l}$ is contained in  $\Path{s}$ and the right bound of $\Path{s}$ will be at the
right
of $c$ in $\Model{\P}$. This property will be preserved in $\Model[']{\P}$. Similarly
if both paths $\Path{s}$ and $\Path{l}$ were overlapping in $\Model{\P}$, they are still
overlapping in $\Model[']{\P}$.

~
\\ \noindent
{(3)}
Since $\Path{z}$ is the right  (\emph{resp.} left) bound of some paths representing some
elements of $M$,
and is not the left (\emph{resp.} right) bound of any other elements of $M$,
there exists a vertex
$b$ in $T$ adjacent to $a$ and such that all the paths representing elements of $M$ that end
at $a$
pass through $b$. To obtain the new representation $\Model[']{\P}$ it suffices
to subdivide this edge one time.
Let $c$ the newly introduced vertex. Then the trivial path $\Path{z}$  in $\Model{\P}$ is
replaced by
a path going from $a$ to $c$. By the transformation, we can observe that all the paths
that were containing $\Path{z}$ in $\Model{\P}$ still contain $\Path{z}$
in $\Model[']{\P}$. Let $s$  be an element of $\P$ such that $\Path{z} \subset \Path{s}$ in
$\Model{\P}$.
If $\Path{s}$ was containing $\Path{z}$ it had to pass through $a$ and $b$, thus by subdividing
$a,b$ we can also conclude that this paths will pass through $a$ and $c$, the added
vertex, in $\Model[']{\P}$.
\end{proof}

\begin{figure}
 \begin{center}
  \includegraphics{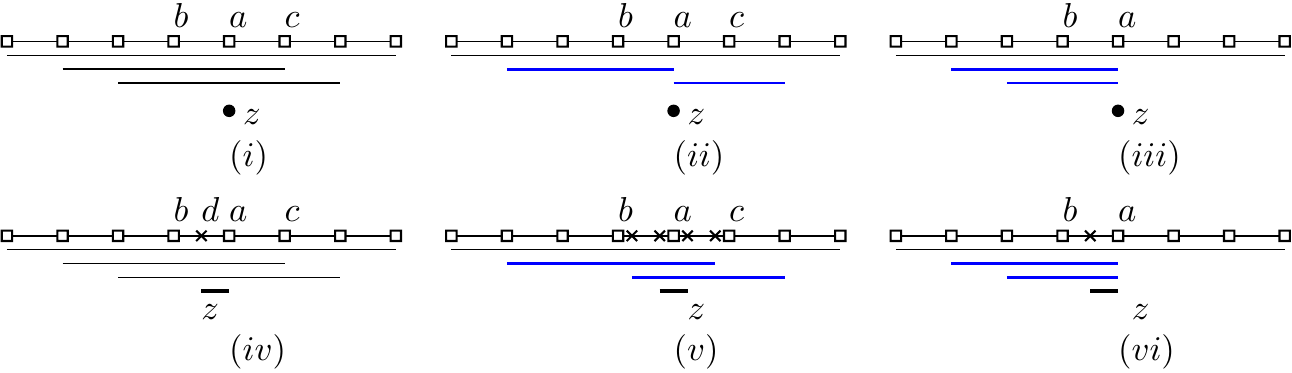}
 \end{center}
 \caption{Representation of prime modules with the element $z$ represented as
 trivial path.}
 \label{fig:PrimeModuleTrivialPath}

\end{figure}

\begin{thm}
 If $\P$ is  dually-CPT  and in a representation $\Model{\P}$ some elements of
 strong modules are represented by trivial paths, then there exists an equivalent
 representation $\Model[']{\P}$ where all the paths representing elements of strong modules
 are non-trivial paths.
\end{thm}

\begin{proof}
 It is a direct consequence of Lemmas  \ref{lem:CliqueModuleTrivialPath},
 \ref{lem:StableModuleTrivialPath} and
 \ref{lem:PrimeModuleTrivialPath} and the fact that each time a trivial path
 is replaced by a  non-trivial one, no trivial path is created in $\Model[']{\P}$.
\end{proof}

From the preceding  theorem, we know
how to obtain a representation a dually-CPT poset where
all the elements contained in non-trivial strong modules
are represented by non-trivial paths. Hence, in
this representation some elements that do not belong
to strong modules might be represented by trivial paths.

\section{Ending of modules onto trivial paths}
\label{sec:ModuleVsTrivPath}

In the previous section we proved that for a dually-CPT poset,
one can always obtain a representation
where no element of a strong module is represented by a trivial path.
It therefore remains to consider
how the paths that represent a strong module $M$ can connect to
an element $z$, not contained in a strong module, where $z$ is represented
by a trivial path in the representation $\Model{\P}$.
Since we need to reconfigure the containment
relation inside the module, this operation could be prevented or constrained if the
trivial path is misplaced. In the case where the trivial path is in the middle of the paths of
the module, it will be easy to reconfigure the containment relation. In the opposite
case, if all the paths representing elements of a module arrive at a trivial path,
we cannot perform the intended operation as planned. In this section, we will identify
the problematic situations, and we will show how to overcome these problems. As in the previous
section, we will perform local changes to the representation to suppress problematic cases.


%
%

When the paths that represent elements of a module are connected to
a trivial path in a representation, several configurations could arise.
The most favorable one, is when the trivial path is properly contained
in the paths of the module (\emph{i.e.} the trivial path
does not lie on any extremity of the path of the module).
Actually this is a configuration we aim at obtaining. The other two
configurations is when all the path have their extremities that
end at a trivial path, or just some of them end at this trivial path.
In most cases we will be able to reconfigure our representation to
obtain a representation that is favorable to our purpose.

\subsection{Complete ending of a module  on a trivial path}

Let us assume that all paths in $\Model{\P}$ corresponding to elements of $M$
have all their extremities end at a vertex $a$ of the host tree $T$.
In that case, there are several possibilities:
   either all the paths that represent $M$ will arrive at $a$ by passing by a vertex
 $b$ of $T$ and such that $a,b$ is an edge of $T$,
 or
 there is another vertex $c$ that is a neighbor of $a$ in $T$ different from $b$ and
 such that some paths of the module pass through $c$.

In this section, even if it is not explicitly stated, the representation
 of the module $M$ will contain the trivial path of $z$ located at the vertex $a$ in $T$.

 \begin{rem}
  If a strong module  of a dually-CPT poset $\P$ is two-sided
  in a representation $\Model{\P}$, then the induced graph
  is not connected. Hence the strong module is a stable module.
 \end{rem}

 \begin{lem}
 Let $M$ be a strong module of a dually-CPT poset $\P$. If $M$ is one-sided
 in a representation \Model{\P} and the poset induced by $M$ is connected,
 then $M$ is a clique module.
 \end{lem}

 \begin{proof}
  If the graph induced by $M$ is connected, $M$ is either a clique module
  or a prime module.
  If $M$ is a clique module, then there is nothing to prove.
  If $M$ is a prime module, then the graph induced by $M$
  necessarily contains an induced $P_4$. Let us show that it is not
  possible to represent a $P_4$ as a CPT representation where all the paths
  end up at a same vertex $a$ of the host tree $T$.
  Consider the representation of the $P_4$ as presented in Figure \ref{fig:P4}$(i)$
  with the containment relation represented in Figure \ref{fig:P4}$(ii)$.

  For a contradiction, let us assume that such a representation exists.
    Since the paths of $2$ and $4$ have to contain the path of $3$ and all these
  paths have to arrive at vertex $a$ of the host tree, we have a configuration
  similar to the one depicted in Figure \ref{fig:P4}$(iii)$ and a part of
  the host tree is depicted in Figure \ref{fig:P4}($iv$).
  Since $2$ and $4$ are not connected, their paths have to diverge in $T$.
  Call  $x$ the vertex of $T$ where these paths diverge.
It remains to represent the path of $1$. Since $1$ is connected to $2$ but not to
$4$,  call $y$ the vertex of $T$  where the path of $1$ begins.
The vertex $y$ has to lie in the proper part of the path of $2$
(see Figure \ref{fig:P4}($iv$)), and this path, by hypothesis, has to go all the way to
$a$. But in that case it has to contain the path of $3$, hence there is a contradiction.
\end{proof}

 \begin{figure}
  \begin{center}
   \includegraphics{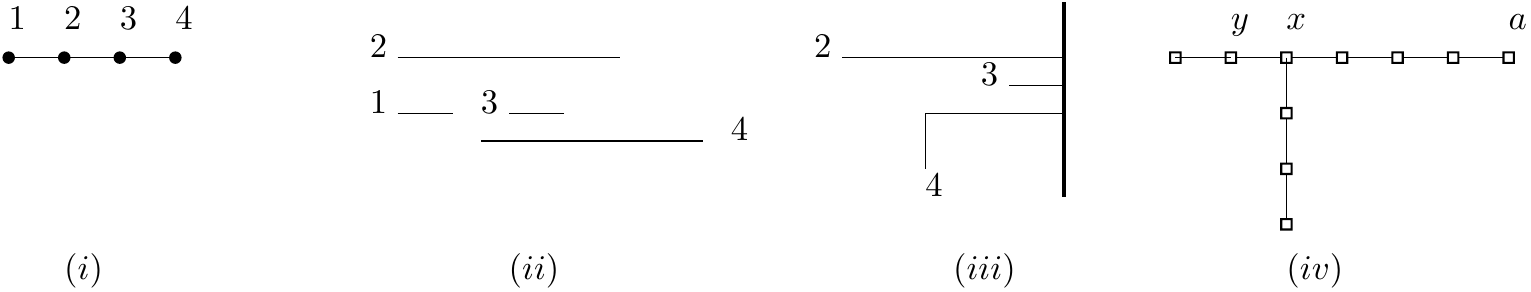}
  \end{center}

  \caption{$(i)$ a $P_4$, $(ii)$ a CI representation of $P_4$, $(iii)$ tentative representation
  with
  all the paths arriving at a vertex, $(iv)$ the host tree of the tentative
  representation.\label{fig:P4}}
 \end{figure}

 We have proven that if in the representation of a strong module all its paths
 arrive at a same vertex of the host tree, then the module is either a clique or a stable module.
 We now consider in which cases can we obtain an alternative representation
 where all the paths do not arrive at a same vertex of the host tree.
When the modification is possible, we will show how by starting from $\Model{\P}$ one
can obtain an equivalent representation $\Model[']{\P}$, that is,
a containment representation that still corresponds to $\P$.

\begin{lem}
\label{lem:ModNotIncluded}
 Let $M$ be a strong module of a dually-CPT poset $\P$ and $\Model{\P}$
  a representation of
 $\P$ where all the paths of $M$ arrive at a same vertex. If there is no
 element of $\P$ that contains the elements of $M$, then there exists
 an alternative representation $\Model[']{\P}$ of $\P$ where all the
 paths of $M$ will have different endpoints.
\end{lem}

\begin{proof}
Let us assume that all the paths of a strong module $M$ arrive at a vertex
$a$ in the representation $\Model{\P}$. If there is no element of $\P\setminus M$ that
contains all the paths of $M$, then we can add a new branch to the host tree
starting at $a$ and ending at $b$ (see Figure \ref{fig:ModNotIncluded}).
Let us denote by $k$ the cardinality of $M$. In order to guarantee that all paths
end on a dedicated vertex, the new branch needs to have at least $k$ new
vertices. It is easy to make sure that the containment relation inside the module
is not altered in this new representation. It is simple to notice that
the previous containments of $\P$ are preserved by this modification and
no new containment is added since the branch only contains paths of $M$.
\end{proof}

\begin{figure}

 \begin{center}
  \includegraphics{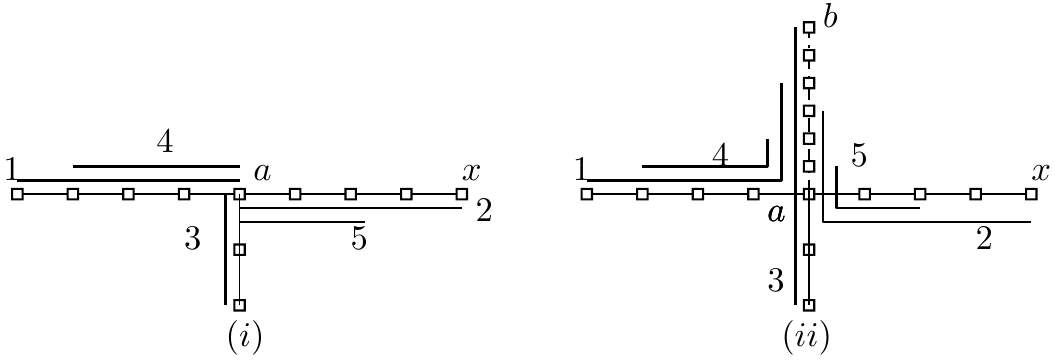}
 \end{center}

 \caption{Example of modification on a representation of a poset $\P$}
 \label{fig:ModNotIncluded}
\end{figure}

  We now consider the case when there is at least one element $x$
 not in $M$ that is greater than all
 the elements of $M$.
In that case $M$ is either one-sided or two-sided.
Let us start with this second case.

\begin{lem}
\label{lem:Stable-Path}
 Let $M$ be a strong stable module of a dually-CPT poset $\P$
 and let $x$ be an element of
 $\P\setminus M$ that contains all the elements of $M$.
 Let us assume that in a representation $\Model{\P}$ all the
 paths of $M$ arrive at a vertex $a$ of $T$. Then there
 exists an equivalent representation $\Model[']{\P}$ where all
 the endpoints of the paths of $M$ near $a$ are distinct.
\end{lem}

\begin{proof}
 By hypothesis, since the elements of $M$ are all contained in an element $x$
 of $\P$, it means that in any representation $\Model{\P}$ of $\P$ the
 union of the paths of $M$ is a path. If in a representation
 $\Model{\P}$ of $\P$ all the paths of $M$ arrive at $a$,
 let $b$ and $c$ be the
 immediate neighbors of $a$ on $T$ along the path that hosts all the paths of $M$.
 Since the strong module considered is stable and in the representation
 every element lies under the path of $x$, the module is two-sided at $a$.
 Since $M$ is two-sided in the representation, its elements
 can be partitioned into two sets $B$ and $C$ as follows:
 An element $r$ is in $B$ if its path in $\Model{\P}$ passes by the vertex $b$.
 Similarly, an element $s$ is in $C$ if its path in $\Model{\P}$ passes by $c$
 (see Figure \ref{fig:Stable2-sided}).
 To obtain $\Model[']{\P}$ it suffices
 to subdivide the edges $a,b$ and $a,c$ of $T$. All
 the paths of the elements of $B$ that previously ended at $a$
 will now end between $a$ and $c$. Hence it is necessary to add $|B|$ new
 vertices between $a$ and $c$. In a symmetric manner, the paths of
 the elements of $C$ will be elongated to end on a new vertex between
 $a$ and $b$; thus it is necessary to add $|C|$ new vertices between
 $a$ and $b$. It is simple to see that the introduced modification
 does not alter the containment relationship. Any path that contained
 all the elements of $M$ will still contain all the elements
 of $M$. And any path that crossed the section of tree spanned
 by the elements of $M$ but did not contain them, will still not
 contain them.
 \end{proof}

\begin{figure}
 \begin{center}
  \includegraphics{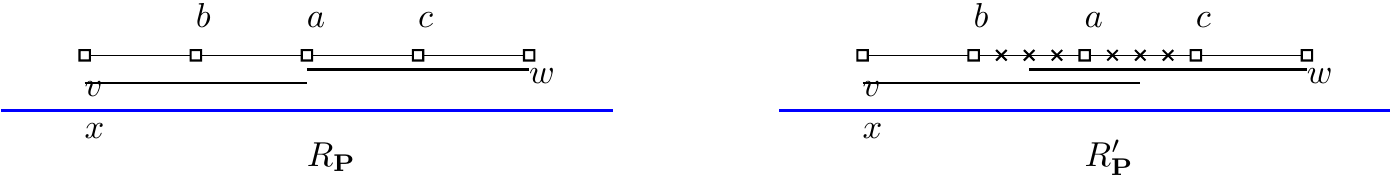}
 \end{center}

\caption{Modification of  the representation of a two-sided stable module.}
\label{fig:Stable2-sided}
\end{figure}

\begin{lem}
\label{lem:Clique-Proper}
Let $M$ be a strong clique module of a dually-CPT poset $\P$ and let $x$ be an element of
 $\P\setminus M$ that contains all the elements of $M$.
 If in a representation $\Model{\P}$ all the paths of $M$ arrive
 at a vertex $a$, then $M$ does not contain any other strong module.
\end{lem}

\begin{proof}
Because of the element $x$, the union of all the
paths of the elements of $M$ in $\Model{\P}$ is included in
the path of $x$ and hence itself forms a path. Since all these paths
are bounded at $a$, then for any pair  of elements $p$ and $q$  of $M$
either the path of $p$ is contained in the path of $q$ or the converse.
There is no pair of non-adjacent vertices. As a consequence
it does not contain any other module.
\end{proof}

\begin{figure}
 \begin{center}
  \includegraphics[width=\textwidth]{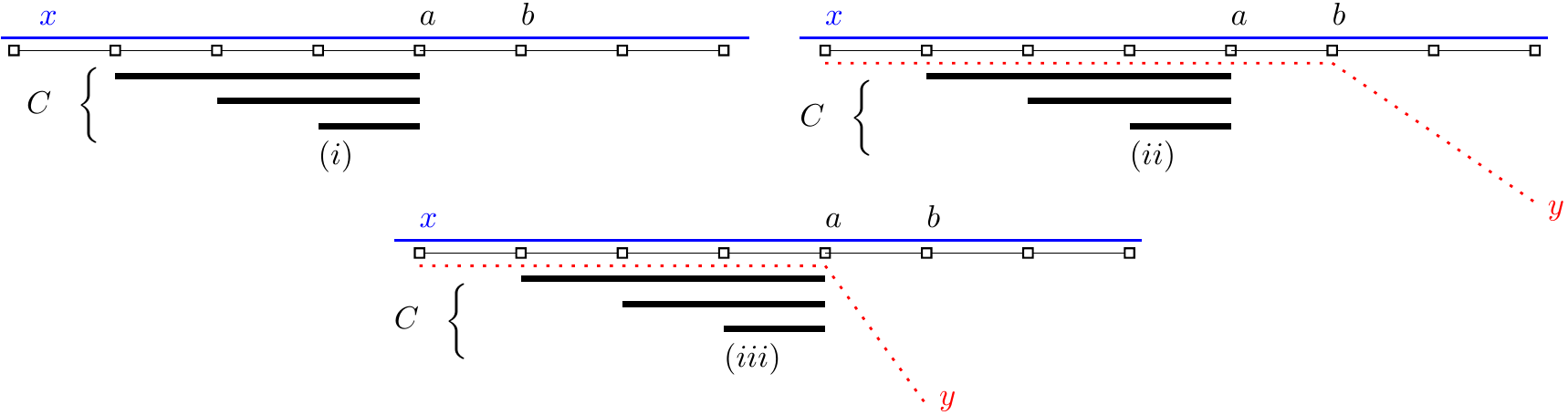}
 \end{center}
 \caption{Configuration of a clique path. The set of paths $C$ represents the paths of a strong clique module $M$}
 \label{fig:Clique-Path}

\end{figure}

Let $M$ be a strong clique module with representation $\Model{\P}$
where all the paths of the elements of $M$ stop at a vertex $a$.
We say that $M$ is \emph{free} in $\Model{\P}$ if there is at least one vertex
$b$ of $T$ such that $a,b$ is an edge of $T$, no path of $M$ passes through
$b$ and all the paths that contain the paths of  $M$ pass through $b$.
(See Figure \ref{fig:Clique-Path} $(i)$ and $(ii)$.)

\begin{lem}
\label{lem:Clique-Path}
 Let $M$ be a strong clique module of a dually-CPT poset $\P$,
   and let $x$ be an element of
 $\P\setminus M$ that contains all the elements of $M$.
 If $M$ is free in a representation $\Model{\P}$ where all the paths of $M$ arrive
 at a vertex $a$, then we can find a representation
 $\Model[']{\P}$ where all the endpoints of $M$ arrive on different
 vertices of $T$.
\end{lem}

\begin{proof}
 Since $M$ is free in $\Model{\P}$ we can re-use the technique used in Lemma
 \ref{lem:Stable-Path} by subdividing the edge $a,b$ of $T$.
\end{proof}

Thanks to  Lemmas \ref{lem:ModNotIncluded}, \ref{lem:Stable-Path},
and  \ref{lem:Clique-Path},
we know how modify a representation
in almost all the cases.
However, one case is not covered, namely,
when the module is a clique and it is blocked.
We say that a strong clique module bounded at a vertex $a$ in a representation
$\Model{\P}$ is \emph{blocked} if it is not free.
There are two reasons why $M$ may be blocked:
(1) It might be because a path that contains the elements of $M$
also stops at $a$, or
(2) because there are two elements $x$ and $y$ that contain
all the elements of $M$ and in $\Model{\P}$  their corresponding paths diverge
at $a$. (See Figure \ref{fig:Clique-Path}$(iii)$.)

\begin{rem}
 \label{rem:CliqueBlocked}
 Let $M$ be a strong clique blocked module of a dually-CPT.
 From Lemma \ref{lem:Clique-Proper}, we know that it does not contain
 any other strong module. Hence, a reconfiguration
 of this subposet is just a matter of relabelling the elements.
\end{rem}

\subsection{Partial ending of a module on a trivial path}
In the previous section, we proved that whenever a module is connected to
 an element $z$ of $\P$ represented by a trivial path in a representation
$\Model{\P}$ and all paths that represent the element of $M$ end at this path, we
can either alter the representation to ensure that all the paths
do not end on that trivial path or, the module is a clique and does not
contain any other modules. Hence  it is possible to alter the containment
relation.

If, in the completely opposite direction, a module $M$ is connected to an element
$z$ represented by a trivial path, but no path that represents an  element of
$M$ ends at this trivial path, it does not create any problem to change the
containment relation of the module.

The last case to consider is when $M$ is connected to a trivial
path, but only some paths of $M$ (not all) end at this trivial path.
We will prove that in that case an equivalent representation, where no path
of $M$ ends at this trivial path, can be obtained.

\begin{lem}
\label{lem:PartialEndingModule}
 Let $M$ be a strong  module of a dually-CPT poset $\P$ connected to an element $z$ ($z\notin M$).
 If in a representation $\Model{\P}$ the element $z$ is represented by
 a trivial path $\Path{z}$ and the paths of some elements of $M$ end at the path
 of $\Path{z}$ and some other paths of elements of $M$
 properly contain $\Path{z}$, then there exists an
 equivalent representation $\Model[']{\P}$  where no element of $M$ ends
 at a trivial path.
\end{lem}

\begin{proof}
Let $I$ denote the set of elements not in $M$ such
that the paths of the elements of $I$ are contained in
the paths of the elements of $M$. In the representation
$\Model{\P}$  all the paths that represent the element of $I$
are all contained in $\cap_{m\in M} P_m$.

Since by hypothesis not all the paths of $M$ end  at
a trivial path, if there are some elements of $M$
that end paths represented by trivial paths, there
are at most two trivial paths in that situation.
Call these trivial paths $y$ and $z$.

Let us assume that the part common to all the paths of $M$
in $\Model{\P}$ is on a horizontal line, and that
\emph{w.l.o.g.}~that $\Path{y}$ is the leftmost
and $\Path{z}$  is the rightmost of this common part.
We assume further, in the
representation $\Model{\P}$, that $\Path{z}$ lies on vertex $a$ of $T$
and $\Path{y}$ lies on vertex $b$ of $T$.

We denote by $L$ (resp.~$R$) the set of all elements of $M$
whose paths in $\Model{\P}$ end at $b$ (resp.~at $a$.) Note that
there is at most one element of $M$ that belongs
to both $L$ and $R$, since the containment relation is proper.

There are two cases to consider:  (1) either there is no element
$x$ such that all the paths of $M$ are contained in the
path of $x$, or (2) such an element $x$ exists.

~\\(1) For the first case, let us assume that such an element
does not exists. Hence there is no path in $\Model{\P}$ that contains
any path of the elements of $M$. In that case, to obtain
an equivalent representation, in $T$ we can add one path
with $|M|$ new vertices connected to $a$ and another path
 with $|M|$ new vertices connected to $b$. Since the
 poset induced by $M$ is CI, it suffices to represent
 this module as a containment of intervals using these
 new branches for the endpoints. The transformation
 process is presented in Figure \ref{fig:PartialEndingCase1}.

The containment relation between elements of $R$ (resp. $L$)
and $I$ remain unchanged. Moreover, for any element $q$ not connected
to $M$, since the endpoints of the paths of the elements of $M$
have been relocated in the two new branches, there is no
containment relation between $\Path{q}$  and the paths
of the elements of $M$, since $\Path{q}$ does not contain any of the new
branches in $\Model[']{\P}$.

~\\(2) Let us now consider the case when there is an element $x$
not in $M$ such that in $\Model{\P}$ the path of $x$ contains all
the paths of the elements of $M$. In the host tree $T$ we denote
by $c$ the neighbor of $a$ such that no path of $R$ passes through $c$
but some paths of elements of $M$ do (by our initial hypothesis).
Let $d$ be the  neighbor of $b$ in $T$ such that paths of some
elements of $M$ pass through but no element of $L$ does.

To obtain an alternative representation $\Model[']{\P}$
we subdivide the edge $a,c$  $|R|$ times  and subdivide the edge $b,d$ $|L|$ times.
  (This transformation
is presented in Figure \ref{fig:PartialEndingCase2}). Then it is just
a matter of extending the paths of the elements in $R$ such
that  they end on a vertex located between $a$ and $c$.
For each element of $R$, its new ending vertex is determined
according to the containment relation in $R$.
For the elements of $L$, we proceed in a similar manner.

It remains to prove that the new representation still represents the
poset $\P$. The only paths that are transformed are the paths that
correspond to elements of $R$ and $L$. Without loss of generality, let $l$ be an element of $R$
and let
$\Path{l}$ be its path in $\Model[']{\P}$.  Since $\Path{l}$ has been extended, it is clear
that all the paths in $\Model{\P}$ that were contained in $\Path{l}$ remain contained
in $\Model[']{\P}$. In addition, since the extension occurred between $a,c$ or $b,d$.
Equivalently let $k$ be an element of $\P$. If $\Path{l} \subset \Path{k}$ in $\Model{\P}$ then
$\Path{l} \subset \Path{k}$ in $\Model[']{\P}$. If $k$ is an element of $R$, by the
transformation
we ensure that the containment relation is preserved. If $k$ is not an element of $R$,
then in $\Model{\P}$, the path $\Path{k}$ passed by vertex $c$ of $T$, hence by extending
 $\Path{l}$, it will not reach $c$, then it is still contained in $\Path{k}$ in
 $\Model[']{\P}$.

Let us now consider an element $q$ such that $\Path{q}$ intersects  $\Path{l}$ but there
is no containment relation in $\Model{\P}$. If $\Path{l} \cup \Path{q}$ is not a path in
$\Model{\P}$
then it contains a claw pattern and this pattern will be preserved
in $\Model[']{\P}$.
Let us now consider the case when $\Path{l} \cup \Path{q}$ forms a path in
$\Model{\P}$.
 If $\Path{q}$ passes through $a$ in $\Model{\P}$ it has one endpoint contained between the
 endpoint
 of $\Path{l}$. Thus the first endpoint of $\Path{q}$  is  at the left of $a$ in $\Model{\P}$
 and
 the  endpoint at the right of $c$ (possibly $c$).  Since $\Path{l}$ does not reach $c$ in
 $\Model[']{\P}$,
 the overlap relation is preserved in the new representation. If both paths were disjoint,
 they remain disjoint in $\Model[']{\P}$.

\end{proof}

\begin{figure}[h!]
 \begin{center}
  \includegraphics{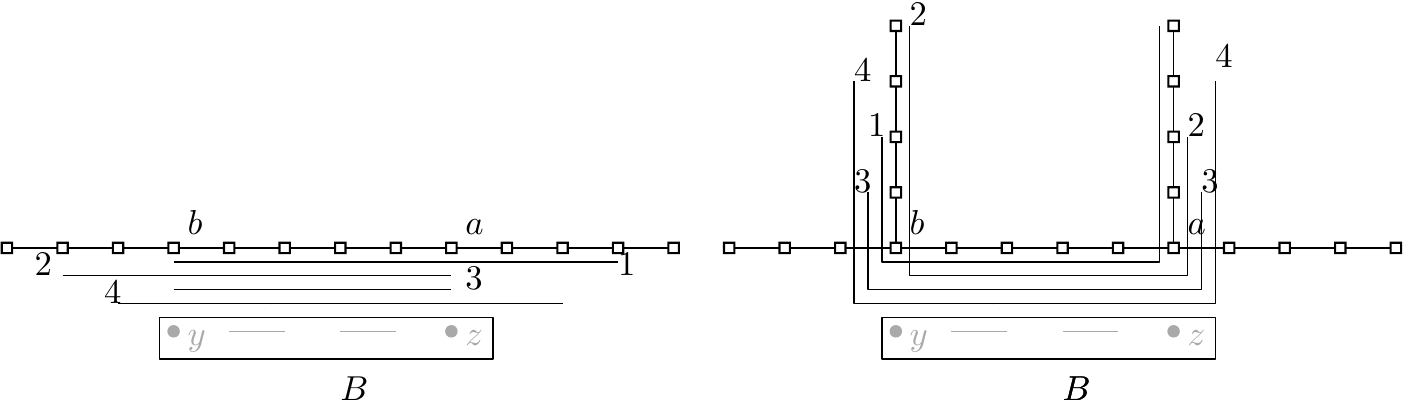}
 \end{center}
 \caption{Illustration of case (1) of Lemma \ref{lem:PartialEndingModule}.
 Elements $1,2,3$ and $4$ are parts of the modules. The module is connected
 to the elements represented by the paths in the box $B$. Elements $1$ and $3$ belong to $L$
 and elements $2$ and $3$ belong to $R$.}
 \label{fig:PartialEndingCase1}
\end{figure}

\begin{figure}[h!]
 \begin{center}
  \includegraphics{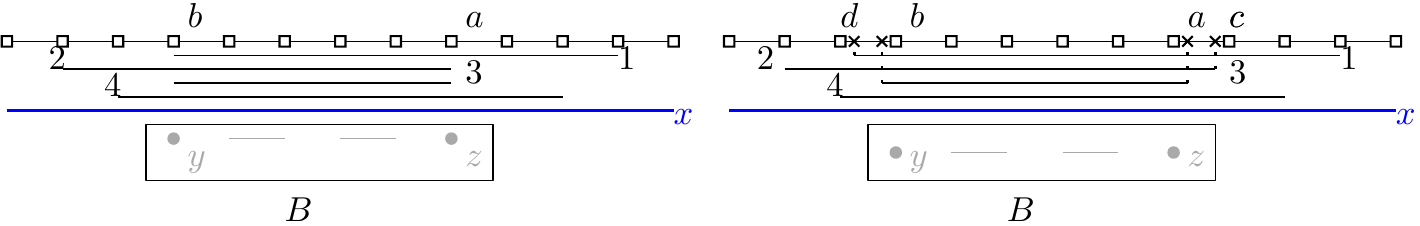}
 \end{center}
 \caption{The same example as in Figure \ref{fig:PartialEndingCase1}, but this time
 there is an element $x$ not in $M$ that contains all the elements of $M$ and that
 prevents performing the modification of  case (1).}
 \label{fig:PartialEndingCase2}
\end{figure}

From Lemmas \ref{lem:ModNotIncluded}, \ref{lem:Stable-Path}, \ref{lem:Clique-Path} and
Remark \ref{rem:CliqueBlocked}, we can summarize the results of this section
with the following theorem:

\begin{thm}
Let $\P$ be a dually-CPT poset. Either for  each strong module $M$
of $\P$ there exists a representation $\Model{\P}$
such that all the paths of $M$ do not end on a trivial
path, or $M$ is a clique blocked module.
 \label{thm:StandardizedRepresentation}

\end{thm}

We call a representation that fulfills the condition of the previous theorem
a \emph{normalized representation}.

\section{Substitution}
\label{sec:Substition}

The last step to obtain our main result is to prove
that for any dually-CPT poset $\P$  all the posets
$\mathcal{Q}=\{\Q_1,\ldots,\Q_l\}$ that are associated to $\P$
admit a CPT representation. Let us consider one particular poset $\Q$
of this set. If $\Q$ is associated to $\P$ it means by definition
that their underlying  comparability graphs are identical. We assume
that  $\P$ is not CI, otherwise the results already stand
from Theorem   \ref{t:strong-property}.
Thus we deduce that the quotient poset of $\P$ is not CI, by
 Theorem \ref{t:dually-nprime}, and thus is prime.
 Since $\P$ and $\Q$ are associated,
 by Property \ref{p:associated-posets} they admit the same
 set of strong modules. The quotient $\H$ of $\P$ is
 obtained by keeping one element of each strong maximal
 module and the quotient  $\K$ of $\Q$ is
 either equal to $\H$ or to its dual  $\H^d$. Let us consider
 that $\H$ is equal to $\K$.

 To obtain a representation for $\Q$, we will use the normalized
 representation $\Model{\P}$ obtained for $\P$. From
 $\Model{\P}$ it is immediate to obtain a representation $\Model{\H}$
 for $\H$ as it suffices to keep one path for each strong
 module of $\P$. In addition, since it is obtained by
 removing paths from a normalized representation,
 we can consider that all the paths that correspond
 to strong modules which are not clique blocked
 modules, do not end on trivial paths of other
 elements. Then we will show
 that for such elements, we can replace this path
 by an arbitrary CI poset. Finally, to obtain
 a CPT representation for $\Q$ it suffices to replace
 each path that is a representative of a strong module,
 by the corresponding CI poset in $\Q$.  For the
 clique blocked modules, as they do not contain
 other strong modules, they correspond to total orders,
 hence the representation can be preserved, but the labelling
 has to be changed to suit the total order in $\Q$.

Let ${v_0}$ be an element of $\H$ that is a representative of some maximal
strong module of $\P$ that is not a clique blocked module.
Let $\Path{v_0}=(x_1,x_2,\ldots,x_k)$ be its path in $\Model{H}$.
We will assume that $k$ is at least $4$. We will show how to replace
$\Path{v_0}$ by a CI poset $\N$. Let $\Model{\N}=\{I_i\}_{1\leq i\leq n}$
be a $CI$ representation of a poset $\N$ whose
vertices are $u_1,u_2,\ldots, u_m$.

Assume that the intervals $I_i$ (subpaths of a path $I$) are non-trivial, 
no two of them share an end vertex
and there is an edge $cd$ of $I$ contained in 
the total intersection of the intervals $I_i$ -- this assumption 
is guaranteed by Proposition \ref{p:compresing-CI-model}.
 Name $a$ and $b$ the end vertices
of the interval union of the intervals $I_i$. Clearly $[c,d]\subset [a,b]$.
We also assume  that
$a$, $b$, $c$ and $d$ are distinct, and that  neither $c$ nor $d$ are end vertices of an
interval $I_i$.

\begin{process} \label{p:substitution}
The process of replacing
in the representation $\Model{\H}$ the path $\Path{v_0}$ by the
intervals $\{I_i\}_{1\leq i \leq n}$ of the representation $\Model{\N}$ consists of:

\begin{itemize}
\item[(i)] subdividing the edges $x_1 x_2$  and $x_{k-1} x_k$ of $T$
by  adding in each one $n-1$ vertices.

\item[(ii)]subdividing the edge $c d$ of $I$ by adding as many vertices as there are in $T$
between $x_2$ and $x_{k-1}$.

\item[(iii)] removing from $\Model{\H}$ the path $W_{v_0}$ and embedding in its place the
intervals of $S$ in such a way that the vertices $a$, $c$, $d$, $b$ and all others between
them match
with the  vertices $x_1$, $x_2$, $x_{k-1}$, $x_k$ and all others between them,
respectively,
as it is shown in Figure \ref{f:grande}.
\end{itemize}
\end{process}

\begin{figure}
  \centering
  \includegraphics[width=15cm]{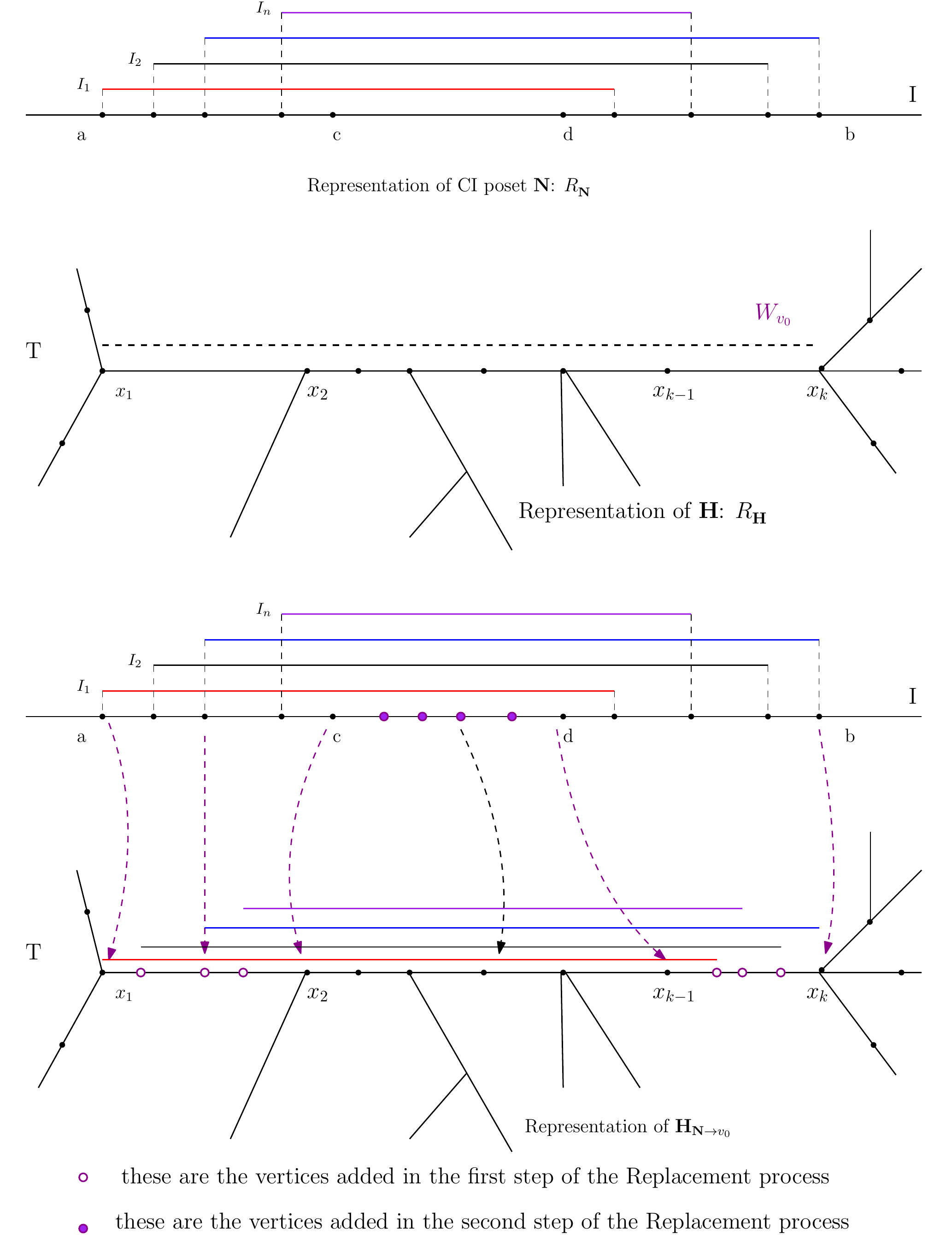}
  \caption{Description of the Replacement process}\label{f:grande}
\end{figure}



\begin{lem}\label{l:substitution-1}
If in $\Model{\H}$ the path $\Path{v_0}$ that represents a module
of a dually-CPT poset $\P$ does not end on trivial paths, then
we can obtain the representation $\Model{\H_{\N\rightarrow v_0}}$
by replacing $\Path{v_0}$ by the
intervals $\{I_i\}_{1\leq i \leq n}$ of the representation $\Model{\N}$
in $\Model{\H_{\N\rightarrow v_0}}$. If
any of the paths $I_{i}$ contains (resp.~is contained in) a path $\Path{v}$, then all
the paths $I_i$ contain
(resp.~are contained in) $\Path{v}$.

Moreover, a path $\Path{v}$ of $\Model{\H}$ contains (is contained in) $\Path{v_0}$
if and only if $\Path{v}$ contains (is contained in) every one of the intervals $I_i$ in
$\Model{\H_{\N\rightarrow v_0}}$.
\end{lem}

\begin{proof}
This result is a direct consequence of two facts: first, that in $\Model{\N}$
no interval $\Path{v}$ of $\Model{\H}$ has an end-vertex between $x_1$ and $x_2$,
nor between $x_{k-1}$ and $x_{k+1}$, and second, that in $\Model{\N}$,
all the intervals $I_i$ contain the interval $x_2x_{k-1}$. See Figure \ref{f:grande}.
\end{proof}

%
%
%
%
\begin{lem}\label{l:substitution-3}
If in $\Model{\H}$ the path $\Path{v_0}$, that represents a blocked clique module
of a dually-CPT poset $\P$,  ends on a trivial path, then
we can obtain the representation $\Model{\H_{\N\rightarrow v_0}}$
by replacing $\Path{v_0}$ by the a collection of paths that
represent a clique.
\end{lem}

\begin{proof}
 Let us assume that $\Path{z}$ is the trivial path that
 $\Path{v_0}$ ends on in $\Model{\H}$. Let us denote
 by $a$ the vertex of the host tree that hosts $\Path{z}$.
 Since the containment relation is proper, we can assume
 that $\Path{v_0}$ passes through at least two vertices of the
 host tree. One of the extremities of $\Path{v_0}$ is $a$. Let us
 call the other extremity $b$. Since the length of $\Path{v_0}$ is
 at least two, we know there exists in the host tree a
 vertex $c$ that is the immediate neighbor of $b$ on
 the path going to $a$. The vertex $c$ is possibly equal
 to $a$. By subdividing an appropiate number of times
 the edge $bc$ of the host tree, we can add as many
 paths as we need to place a clique module. From the
 transformation, it is easy to see that the containment
 relation is preserved with respect to the module.

\end{proof}


We restate here our main theorem:

\begin{thm}
\label{thm:Main2}
 A poset $\P$ is strongly-CPT if and only if it is dually-CPT.
\end{thm}

\begin{proof}

Let $\H=\P/\mathcal{M}(\P)$ be the quotient poset,
where $\mathcal{M}(\P)=\left\{M_1, \ldots, M_k\right\}$ is the maximal modular partition
of $\P$.

Since $\P$ is a dually-$CPT$ poset and $\H$ is a
subposet of $\P$, then $\H$ and $\H^d$ admit a normalized $CPT$-representation.
If $\H$ is $CI$, by Remark \ref{r:strong}  and Theorem \ref{t:dually-nprime}, $\P$ is $CI$ and
so strongly-$CPT$.
Thus let us assume that $\H$ is a prime dually-$CPT$ poset.

Let $\Q$ be an associated poset of $\P$ and let $\K$ be its quotient
poset.
Since $\P$ and $\Q$ are associated, an immediate
consequence is that $\H$ and $\K$ are associated;
in addition by hypothesis they are both prime,
hence by Theorem \ref{t:indecommposable-poset},
$\K$ is either equal to $\H$ or to $\H^{d}$.
Let us assume, \emph{w.l.o.g.}, that $\H=\K$.

We will prove that $\Q$ admits a $CPT$ representation.
By Theorem \ref{t:orientations} and \emph{w.l.o.g}, we assume that $\Q$
is obtained by replacing in $\H$ each vertex $v_i$
of $\H$ for $\Q_i=\Q(M_i)$.
By Proposition \ref{p:associated-posets}, $\P$ and
$\Q$ possess the same strong modules and
by Theorem \ref{t:dually-nprime} since $\P$ is dually-CPT,
all the strong modules of $\P$ and $\Q$ are CI.
For each $\Q_i$ we have a $CI$ representation.

Let $\Model{\H}$ be a $CPT$ representation of $\H$,
obtained from a normalized representation of $\P$.
The representation is obtained by only keeping one
path for each strong module of $\P$.

For each path $\Path{v_i}$ that corresponds to a module
$M_i$ of $\Q$, if $\Path{v_i}$ does not end on a trivial
path of $\Model{\H}$ then it corresponds to a  module
which is not a blocked clique module, hence
by Lemma \ref{l:substitution-1}, we can replace $\Path{v_i}$
by a CI representation of $\Q_{i}$.

The only remaining case is if $\Path{v_i}$ ends on a trivial
path in $\Model{\H}$. In that case, it means that
it corresponds to a blocked clique module of $\P$ in the representation
$\Model{\P}$. Hence by Lemma \ref{l:substitution-3}, we can
replace $\Path{v_i}$ by a CI representation
of the maximal strong clique module $\Q_i$.

By proceeding in that way for each maximal strong module,
we are able to obtain a CPT representation   $\Model{\Q}$ of $\Q$.
\end{proof}


\end{document}